\documentclass[11pt,a4paper]{article}
\pdfoutput=1

\usepackage[dvipsnames,usenames]{xcolor}
\usepackage[colorlinks=true,urlcolor=Blue,citecolor=Green,linkcolor=BrickRed]{hyperref}
\usepackage{amsmath,amsthm,amssymb,amsfonts}
\usepackage[utf8]{inputenc}
\usepackage{todonotes}
\usepackage{xspace}
\usepackage{multicol}
\usepackage{thmtools,thm-restate}
\usepackage{dsfont} %for \ZZ

\usepackage{authblk}
\usepackage{fullpage}
\usepackage{multirow}

\usepackage{tikz}

\usepackage{cellspace}

\newtheorem{theorem}{Theorem}[section]
\newtheorem{lemma}[theorem]{Lemma}
\newtheorem{corollary}[theorem]{Corollary}
\newtheorem{definition}[lemma]{Definition}

\newtheorem{proposition}[lemma]{Proposition}
\newtheorem{example}[lemma]{Example}

\newtheorem{def-restatable}[theorem]{Definition}
\newtheorem{obs-restatable}{Lemma}
\newtheorem{lemma-restatable}[theorem]{Lemma}
\newtheorem{corollary-restatable}[theorem]{Corollary}
\newtheorem*{conjecture}{Conjecture}

\newcommand{\eps}{\varepsilon}
\newcommand{\Oh}{\mathcal{O}}
\newcommand{\D}{\mathcal{D}\,}
\renewcommand{\S}{\mathcal{S}\,}
\newcommand{\I}{\mathcal{I}\,}
\def\polylog{\operatorname{polylog}}

\newcommand{\DontCares}{\textsc{Pattern Matching with Don't Cares}\xspace}
\newcommand{\GPM}{\textsc{GPM}\xspace}
\newcommand{\GPMFull}{\textsc{Generalised Pattern Matching}\xspace}
\newcommand{\countchars}{\mathsf{count}}
\newcommand{\range}{\mathsf{range}}
\newcommand{\epsdep}{\eps^{-2}}

\newcommand{\UBddetnoeps}{\D n \log^6 n}
\newcommand{\timeUBddet}{\Oh(\epsdep \UBddetnoeps)}
\newcommand{\UBsdetnoeps}{\sqrt{\S} n \log^{7/2} n}
\newcommand{\timeUBsdet}{\Oh(\eps^{-1} \UBsdetnoeps)}

   \newcommand{\defproblem}[3]{
  \vspace{2mm}
\noindent\fbox{
  \begin{minipage}{0.96\textwidth}
  #1\\
  {\bf{Input:}} #2  \\
  {\bf{Output:}} #3
  \end{minipage}
  }
  \vspace{2mm}
}

\title{Generalised Pattern Matching Revisited}
\date{}

\author[1]{Bartłomiej Dudek}
\author[1]{Paweł Gawrychowski}
\author[2]{Tatiana Starikovskaya}

\affil[1]{Institute of Computer Science, University of Wrocław, Poland}
\affil[2]{DIENS, \'{E}cole normale sup\'{e}rieure, PSL Research University, France}

\begin{document}
\maketitle

\begin{abstract}  
In the problem of $\textsc{Generalised Pattern Matching}\ (\textsc{GPM})$ [STOC'94, Muthukrishnan and Palem], we are given a text $T$ of length $n$ over an alphabet $\Sigma_T$, a pattern $P$ of length $m$ over an alphabet $\Sigma_P$, and a matching relationship $\subseteq \Sigma_T \times \Sigma_P$, and must return all substrings of $T$ that match $P$ (\emph{reporting}) or the number of mismatches between each substring of $T$ of length $m$ and $P$ (\emph{counting}). In this work, we improve over all previously known algorithms for this problem:

\begin{itemize}
\item For $\mathcal{D}\,$ being the maximum number of characters that match a fixed character, we show two new Monte Carlo algorithms, a reporting algorithm with time $\mathcal{O}(\mathcal{D}\, n \log n \log m)$ and a $(1-\varepsilon)$-approximation counting algorithm with time $\mathcal{O}(\varepsilon^{-1} \mathcal{D}\, n \log n \log m)$. We then derive a $(1-\varepsilon)$-approximation deterministic counting algorithm for $\textsc{GPM}$ with $\mathcal{O}(\varepsilon^{-2} \mathcal{D}\, n \log^6 n)$ time. 

\item For $\mathcal{S}\,$ being the number of pairs of matching characters, we demonstrate Monte Carlo algorithms for reporting and $(1-\varepsilon)$-approximate counting with running time $\mathcal{O}(\sqrt\mathcal{S}\, n \log m \sqrt{\log n})$ and $\mathcal{O}(\sqrt{\varepsilon^{-1} \mathcal{S}\,} n \log m \sqrt{\log n})$, respectively, as well as a ${(1-\varepsilon)}$-approximation deterministic algorithm for the counting variant of $\textsc{GPM}$ running in $\mathcal{O}(\varepsilon^{-1} \sqrt{\mathcal{S}} n \log^{7/2} n)$ time.

\item Finally, for $\mathcal{I}\,$ being the total number of disjoint intervals of characters that match the $m$ characters of the pattern $P$, we show that both the reporting and the counting variants of $\textsc{GPM}$ can be solved exactly and deterministically in $\mathcal{O}(n\sqrt{\mathcal{I}\, \log m} +n \log n)$ time. 
\end{itemize} 

At the heart of our new deterministic upper bounds for $\mathcal{D}\,$ and $\mathcal{S}\,$ lies a faster construction of superimposed codes, which solves an open problem posed in [FOCS'97, Indyk] and can be of independent interest. 

To conclude, we demonstrate first lower bounds for $\textsc{GPM}$. We start by showing that any deterministic or Monte Carlo algorithm for $\textsc{GPM}$ must use $\Omega(\mathcal{S})$ time, and then proceed to show higher lower bounds for combinatorial algorithms. These bounds show that our algorithms are almost optimal, unless a radically new approach is developed. 
\end{abstract}

\newpage
 \pagenumbering{arabic}

\section{Introduction}
\label{sec:introduction}
Processing noisy data is a keystone of modern string processing. One possible approach to address this challenge is approximate pattern matching, where the task is to find all substrings of the text that are close to the pattern under some similarity measure, such as Hamming or edit distance. The approximate pattern matching approach assumes that noise is arbitrary, i.e. that we can delete or replace any character of the pattern or of the text by any other character of the alphabet. 

The assumption that the noise is completely arbitrary is not necessarily justified, as in practice we might have some predetermined
knowledge about the structure of the errors. In this paper we focus on the \textsc{Generalised Pattern Matching} (\GPM) problem that addresses this setting. We assume to be given a text $T$ over an alphabet $\Sigma_T$, a pattern $P$ over an alphabet $\Sigma_P$, and we allow each character of $\Sigma_T$ to match a subset of characters of $\Sigma_P$. We must report all substrings of the text that match the pattern. 
This problem was introduced in STOC'94~\cite{MuthukrishnanP94} by Muthukrishnan and Palem to
provide a unified approach for solving different extensions of the classical pattern matching question that has been considered as separate problems in the early 90s. Later, Muthukrishnan~\cite{Muthukrishnan95} considered a \emph{counting variant} of \GPM, where the task is to count the number of mismatches between substrings of the text and the pattern. Formally, the problem is defined as follows:

\vspace{2mm}
\noindent\fbox{
\begin{minipage}{0.96\textwidth}
  \textsc{Generalised Pattern Matching} (\GPM)\\
  {\bf{Input:}} A text $T \in (\Sigma_T)^n$, a pattern $P \in (\Sigma_P)^m$, and a matching relationship $\subseteq \Sigma_T \times \Sigma_P$. \\
  {\bf{Output (Reporting):}} All $i \in [n-m+1]$ such that $T[i,i+m-1]$ matches $P$.\\
   {\bf{Output (Counting):}} For each $i \in [n-m+1]$, the number of positions $j \in [m]$ such that $T[i+j-1]$ does not match $P[j]$.
  \end{minipage}
  }
  \vspace{2mm}

Muthukrishnan and Palem~\cite{MuthukrishnanP94} and subsequent work~\cite{Muthukrishnan95,MuthukrishnanR95} considered three natural parameters describing the matching relationship ($\D, \S$) or the pattern ($\I$). 
Viewing the matching relationship as a bipartite graph with edges connecting pairs of matching characters from $\Sigma_T \times \Sigma_P$, $\D$ is the maximum degree of a node and $\S$ is the total number of edges in the graph.
Next, the parameter $\I$ describes the pattern rather than the matching relationship. 
For each character $a \in \Sigma_P$, let $I(a)$ be the minimal set  of disjoint sorted intervals that contain the characters that match~$a$, and define $\I = \sum_{j \in [m]} |I(P[j])|$.

\vspace{-0.4cm}
\paragraph{The maximum number of characters that match a fixed character, $\D$.} For the reporting variant of \GPM, Muthukrishnan~\cite{Muthukrishnan95} showed a Las Vegas algorithm with running time $\Oh(\D n  \log n \log m)$. Indyk~\cite{I:1997} used superimposed codes to show a deterministic algorithm with running time $\Oh(|\Sigma_P| \D^2 \log^2 n  + \D n \log^3 n \log m)$. For the counting variant, Muthukrishnan~\cite{Muthukrishnan95} showed a $(\log m)$-approximation Las Vegas algorithm with time $\Oh(\D n \log n \log m)$. 
Indyk~\cite{I:1997}  gave a $(1-\eps)$-approximation deterministic and Monte Carlo algorithm with running time
$\Oh (\eps^{-2}\D^2 n \log^3 n)$ and $\Oh (\eps^{-2} \D n \log^3 n)$, respectively.

\vspace{-0.4cm}
\paragraph{The number of matching pairs of characters, $\S$.} Muthukrishnan and Ramesh~\cite{MuthukrishnanR95} gave an $\Oh((\S m \log^2 m)^{1/3} n)$-time algorithm for the reporting variant of \GPM. 

\vspace{-0.4cm}
\paragraph{The number of intervals of matching characters, $\I$.} For this parameter, Muthukrishnan~\cite{Muthukrishnan95} gave an $\Oh(\I + (m \I)^{1/3} n \sqrt{\log m})$-time algorithm\footnote{\cite[Theorem 9]{Muthukrishnan95} claims $\Oh(n+\I+\I^{1/3} (nm)^{2/3} \sqrt{\log m})$, but the first sentence of the proof states that for $n \le 2m$ the algorithm takes $\Oh(\I + \I^{1/3} m^{4/3} \sqrt{\log m})$ time, where the first term is the time that we need to read the input. For a longer text, one needs to apply it $n/m$ times for overlapping blocks of length $2m$, making the total time $\Oh(\I + n/m \cdot \I^{1/3} m^{4/3} \sqrt{\log m}) = \Oh(\I + (m \I)^{1/3} n \sqrt{\log m})$.}.

\subsection{Our Contribution}
We improve existing randomised and deterministic upper bounds for \GPM, and demonstrate matching lower bounds. At heart of our deterministic algorithms for the counting variant of \GPM is a solution to an open problem of Indyk~\cite{I:1997} on construction of superimposed codes.

\paragraph{Data-dependent superimposed codes.}
A $z$-superimposed code is a set of binary vectors such that no vector is contained in a Boolean sum (i.e. bitwise OR) of $z$ other vectors. Superimposed codes find their main application in information retrieval (e.g. in compressed representation of document attributes), and optimizing broadcasting on radio networks~\cite{Kautz06}, and have also proved to be useful in graph algorithms~\cite{DBLP:conf/icalp/GeorgiadisGIPU17,DBLP:journals/corr/abs-1807-05803}. Indyk~\cite{I:1997} extended the notion of superimposed codes to the so-called \emph{data-dependent superimposed codes}, and asked for a deterministic construction for such codes with a certain additional property that makes them useful for counting mismatches (see Section~\ref{sec:superimposed_codes} for a formal definition). We provide such a construction algorithm in Theorem~\ref{th:superimposed_codes}. 
We briefly describe the high-level idea below.

We need the concept of discrepancy minimization. Given a universe $U$, each of its elements is assigned one of two colours, red or blue. The \emph{discrepancy} of a subset of $U$ is defined as the difference between the number of red and blue elements in it, and the discrepancy of a family $\mathcal{F}$ of subsets is defined as the maximum of the absolute values of discrepancies of the subsets in $\mathcal{F}$. Discrepancy minimization is a fundamental notion with numerous applications, including derandomization, computational geometry, numerical integration, understanding the limits of models of computation, and so on (see e.g.~\cite{Chazelle01}).
A recent line of work showed a series of algorithms for constructing colourings of low discrepancy in various settings~\cite{Lovett15,Bansal10,BansalS13,BansalCKL14,BansalG17,BansalDGL18,BansalDG19,Larsen19}. 
For our applications, we need to work under the assumption that the size of each subset in $\mathcal{F}$ is bounded by a given parameter $k$.
In Theorem~\ref{th:discrepancy}, we describe a fast deterministic algorithm that returns a colouring of small discrepancy for this case.
We follow the algorithm described by Chazelle~\cite{Chazelle01} that can be roughly summarized as based on the method of conditional expectations tweaked
as to allow for an efficient implementation. In more detail, Chazelle's construction assumes infinite precision of computation and does not
immediately translate into an efficient algorithm working in the Word RAM model of computation, thus requiring resolving some technical
issues to bound the required precision and the overall complexity.

We apply discrepancy minimization to design in Lemma~\ref{lm:main_hash_function} a procedure that, given a family $\mathcal{F}$ of subsets of $U$,
partitions the universe $U$ into not too many parts such that the intersection of each part and each of the subsets in $\mathcal{F}$ is small.
The procedure follows the natural idea of colouring the universe with two colours, and then recursing on the elements with the
same colour. Every step of such construction introduces some penalty that needs to be carefully controlled as to guarantee the desired property
in the end. Because of this penalty, we are only able to guarantee that the intersections are small, but not constant.
To finish the construction, we combine the partition with a hash function into the ring of polynomials.
We stress that this part of the construction is new and not simply a modification of Chazelle's (or Indyk's) method.

\paragraph{Upper bounds for \GPM.} Similar to previous work, we assume that the alphabets' sizes are polynomial in $n$ and that the matching relationship is given as a graph $M$ on the set of vertices $\Sigma_T\cup\Sigma_P$. We also assume to have access to three oracles that can answer the following questions in $\Oh(1)$ time: 

\begin{enumerate}
\item Is there an edge between $a\in\Sigma_{T}$ and $b\in\Sigma_{P}$ (in other words, do $a$ and $b$ match)?
\item What is the degree of a character $a\in\Sigma_{T}$ or $b\in\Sigma_{P}$ (in other words, what is the number of characters that match a given character)?
\item What is the $k$-th neighbor of $a\in\Sigma_T$ (in other words, what is the $k$-th character $b\in\Sigma_P$ matching $a$)? We assume
an arbitrary (but fixed) order of neighbors of every node.
\end{enumerate}

\noindent
Under these assumptions, we show the following upper bounds summarized in Tables~\ref{tb:GPM-r} and~\ref{tb:GPM-c}:

\begin{enumerate}
\item We start by showing a new Monte Carlo algorithm for the parameter $\D$ with running time $\Oh(\D n \log m \log n)$ (Theorem~\ref{th:UB_d_rand}). While its running time is the same as that of~\cite{Muthukrishnan95}, it encapsulates a novel approach to the problem that serves as a basis for other algorithms. We then derive a Monte Carlo algorithm for the parameter $\S$ with running time $\Oh(\sqrt\S n \log m \sqrt{\log n})$ (Theorem~\ref{th:UB_s_rand}). As a corollary, we show a $(1-\eps)$-approximation Monte Carlo algorithm that solves the counting variant of \GPM in time $\Oh(\min\{\eps^{-1} \D \log n, \sqrt{\eps^{-1} \S \log n}\} \cdot n \log m)$ (Corollary~\ref{cor:UB_s_d_rand_counting}). All three algorithms have inverse-polynomial error probability. 

\item Next, using the data-dependent superimposed codes, we construct $(1-\eps)$-approximation deterministic algorithms for the counting variant of \GPM. The first algorithm requires $\timeUBddet$ time (Theorem~\ref{th:UB_d_det}), and the second algorithm $\timeUBsdet$ time (Theorem~\ref{th:UB_s_det}). By taking $\eps = 1/2$, we immediately obtain deterministic algorithms for the reporting variant of the problem with the same complexities. 

\item Finally, we show that both the reporting and the counting variants of \GPM can be solved exactly and deterministically in $\Oh(n\sqrt{\I \log m} +n \log n)$ time (Theorem~\ref{th:det_i}). 
\end{enumerate}

\begin{table}[t]
\setlength\cellspacetoplimit{5mm}
\setlength\cellspacebottomlimit{5mm}
\begin{center}
\begin{tabular}{|l|l|l|} 
\hline
\textbf{Time} & \textbf{Det./Rand.}  &\\
\hline 
\hline
$\Oh (|\Sigma_P| \D^2 \log^2 n  + \D n \log^3 n \log m)$ & Det. & \cite{I:1997}\\
$\Oh(\UBddetnoeps)$ & Det. & This work\\
$\Oh(\D n  \log n \log m)$ & Rand. & \cite{Muthukrishnan95} \\
$\Oh(\D n \log n \log m)$ & Rand. & This work\\
 \hline
$\Oh((\S m \log^2 m)^{1/3} n)$ & Det. & \cite{MuthukrishnanR95}\\
$\Oh(\UBsdetnoeps)$ & Det. & This work\\
$\Oh(\sqrt{\S} n \log m \sqrt{\log n})$ & Rand. & This work\\
\hline
$\Oh(\I + (m \I)^{1/3} n \sqrt{ \log m})$ & Det. & \cite{Muthukrishnan95}\\
$\Oh(n \sqrt{\I \log m} + n \log n)$ & Det. & This work\\
\hline
\end{tabular}
\end{center}
\caption{\GPMFull (reporting)}
\label{tb:GPM-r}
\end{table}

\begin{table}[t]
\begin{center}
\begin{tabular}{|l|l|l|l|} 
\hline
\textbf{Time} & \textbf{Det./Rand.}  & Approx. factor&\\
\hline 
\hline
$\Oh (\eps^{-2}\D^2 n \log^3 n)$ & Det. & $(1-\eps)$ & \cite{I:1997}\\
$\timeUBddet$ & Det.  & $(1-\eps)$ & This work\\
$\Oh (\D n \log n \log m)$ & Rand. & $\log m$ & \cite{Muthukrishnan95}\\
$\Oh (\eps^{-2}\D n \log^3 n)$ & Rand. & $(1-\eps)$ & \cite{I:1997}\\
$\Oh(\eps^{-1 }\D n \log n \log m)$ & Rand. & $(1-\eps)$ & This work\\
 \hline
$\timeUBsdet$ & Det. & $(1-\eps)$  & This work\\ 
$\Oh(\sqrt{\eps^{-1}\S} n \log m \sqrt{\log n})$ & Rand.  & $(1-\eps)$ & This work\\
\hline
$\Oh(\I + (m \I)^{1/3} n \sqrt{ \log m})$ & Det. & --- & \cite{Muthukrishnan95}\\
$\Oh(n \sqrt{\I \log m} + n \log n)$ & Det.  & --- & This work\\
\hline
\end{tabular}
\end{center}
\caption{\GPMFull (counting)}
\label{tb:GPM-c}
\end{table}

\paragraph{Lower bounds for \GPM.}
We also show first lower bounds for \GPM (see Appendix~\ref{sec:lb}). We start with a simple adversary-based argument that shows that any deterministic algorithm or any Monte Carlo algorithm with constant error probability that solves \GPM must use $\Omega(\mathcal{S})$ time (Lemma~\ref{lm:det_LB} and~\ref{lm:random_LB}). We then proceed to show higher lower bounds for combinatorial algorithms by reduction from Boolean matrix multiplication\footnote{It is not clear what combinatorial means precisely. However, FFT and Boolean convolution often used in algorithms on strings are considered not to be combinatorial.} parameterized by $\D, \S, \I$ (Lemma~\ref{lm:cond_LB_s} and Corollary~\ref{cor:cond_LB_d_i}). 
All the lower bounds are presented for the reporting variant of \GPM, so they immediately apply also to the counting variant.
These bounds show that our algorithms are almost optimal, unless a radically new approach is developed. 

\subsection{Related Work}

\paragraph{Degenerate string matching.} A more general approach to dealing with noise in string data is degenerate string matching, where the set of matching characters is specified for every position of the text or of the pattern (as opposed to every character of the alphabets).
Abrahamson~\cite{Abrahamson87} showed the first efficient algorithm for a degenerate pattern and a standard text.
Later, several practically efficient algorithms were shown~\cite{Navarro:2001:NFF:511281.511284,Holub:2008:FPI:1346365.1346680}. 

\paragraph{Pattern matching with don't cares.}
In this problem, we assume $\Sigma_T = \Sigma_P = \Sigma$, where~$\Sigma$ contains a special character~--- ``don't care''. We assume that two characters of $\Sigma$ match if either one of them is the don't care character, or they are equal. The study of this problem commenced in~\cite{Fischer:1974:SOP:889566}, where a $\Oh(n \log m \log |\Sigma|)$-time algorithm was presented. The time complexity of the algorithm was improved in subsequent work~\cite{CH:2002,I:1998,Kalai:2002:EPD:545381.545468}, culminating in an elegant $\Oh(n \log m)$-time deterministic algorithm of Clifford and Clifford~\cite{CLIFFORD200753}. Clifford and Porat~\cite{CLIFFORD20101021} also considered the problem of identifying all alignments where the number of mismatching characters is at most~$k$.

\paragraph{Threshold pattern matching.} In the threshold pattern matching problem, we are given a parameter $\delta$, and we say that two characters $a, b$ match if $|a-b| < \delta$. The threshold pattern matching problem has been studied both in reporting and counting variants~\cite{ATALLAH2011674,cd0a23419254469aa9e7ac8fd7ef9f75,Cantone:2005:EA9:2154763.2154802,10.1007/11496656_7,COLE2003227,Crochemore:2002:OSH:958824.958826,doi:10.1142/S0129054108005607,ZHANG201721}. 
The best algorithm for the reporting variant of the threshold pattern matching problem is deterministic and takes linear time (after the pattern has been preprocessed). The best deterministic algorithm for the counting variant of threshold pattern matching has time $\Oh(( \log \delta + 1) n \sqrt{m \log m} )$, while the best randomised algorithm has time $\Oh(( \log \delta + 1) n \log m)$~\cite{ZHANG201721}.

In threshold pattern matching the matching relationship is described with a single interval per character, so $\I=m$.
Hence from Theorem~\ref{th:det_i} immediately follows a faster deterministic algorithm for the counting variant of the threshold pattern matching problem (Corollary~\ref{cor:threshold}).

\section{Data-Dependent Superimposed Codes}	
\label{sec:superimposed_codes}
\newcommand{\assumptions}{Given a family of $z$ sets $S_i\subseteq U$ where $|S_i|\leq k$ and $|U|=zk$}

We start by solving an open problem posed by Indyk~\cite{I:1997}: provide a deterministic algorithm for construction of a variant of data-dependent superimposed codes that is particularly suitable for the counting variant of \GPM. The solution that we present is rather involved, a reader more interested in pattern matching applications can skip this section on the first reading.

\begin{restatable}{def-restatable}{codes}\label{def:codes}
Let $S_{1},\ldots,S_{z}$ be subsets of a universe $U$. A family of sets $\mathcal{C}= \{C_{1}, \ldots, C_{|U|}\}$, where $C_{u} \subseteq [\ell]$ and $|C_{u}|=w$ for $u\in U$ is called an $(\{S_i\}, \tau)$-superimposed code if for every $S_i$ and $u \notin S_i$ we have $|C_{u}-\bigcup_{v\in S_i} C_{v}| \geq \tau$. We call $\ell$ and $w$ respectively the length and the weight of the code $\mathcal{C}$.
\end{restatable}

Suppose that the size of each $S_i$ is at most $k$, where $k$ is some fixed integer. Indyk asked if there exists a deterministic $\tilde{\Oh}((zk) /\eps^{\Oh(1)})$-time algorithm that computes an
$(\{S_i\}, (1-\eps) w)$-superimposed code of some weight $w$ and length $\ell=\Oh(k\polylog(zk))$. It can be seen that we cannot hope to construct such a code with $\ell$ independent of $\eps$.
In the following lemma we show that even if we restrict to the case of $k=1$ we still need that $\ell$ significantly depends on~$\eps$.

\begin{lemma}
 For every constant $\delta\in(0,1)$, function $f(z) =\Oh(\polylog z)$, and $z$ large enough, there exists a family of singleton sets $S_1,S_2,\ldots,S_z$ and $0 < \eps < 1$ such that any $(\{S_i\},(1-\eps)w)$-superimposed code of weight $w$ must have length length $\ell > f(z)/\eps^{\delta}$.
\end{lemma}
\begin{proof}
 Consider sets $S_i=\{i\}$ for $i\in[z]$, where $z$ will be determined later.
 Let $\eps=1/(2f(z))^{\frac{1}{1-\delta}}$ and suppose that there is a $(\{S_i\},(1-\eps)w)$-superimposed code $\mathcal{C}$. 
 Then, by definition of superimposed codes and from $w\leq \ell$, for $i\ne j$ it holds 
 \[\left\{
\begin{array}{l}
|C_i-C_j|\geq (1-\eps)w=w-\eps w \geq w-\eps \ell,\\
\eps \ell\leq \eps f(z)/\eps^\delta = \eps^{1-\delta} f(z)=1/2
\end{array} 
 \right.
 \]
 so $|C_i-C_j|>w-1$.
 Hence, $|C_i-C_j|=w$ and every $C_i$ and $C_j$ must be disjoint, so $\ell\geq zw\geq z$.
Assume towards a contradiction that $\ell\leq f(z)/\eps^{\delta}$. We obtain
 \[\ell\leq f(z)/\eps^{\delta}=f(z)\cdot(2f(z))^{\frac \delta{1-\delta}}=f(z)^{\frac 1{1-\delta}} \cdot 2^{\frac \delta{1-\delta}}=\Oh(\polylog z)\cdot 2^{\frac \delta{1-\delta}} <z \]
 where the last inequality holds for sufficiently large $z$.
 This leads to contradiction and the claim follows.
\end{proof}

Therefore, one should allow $\ell=\Oh(k\polylog(zk) / \eps^{\Oh(1)})$. We give a positive answer to this natural relaxation. We start by showing an efficient deterministic algorithm for discrepancy minimization that will play an essential role in our approach.

\subsection{Discrepancy Minimization}
Let us start with a formal definition of discrepancy. 

\begin{definition}[Discrepancy]
Consider a family $\mathcal{F}$ of $z$ sets $S_i \subseteq U$, $i \in [z]$. We call a function $\chi: U \rightarrow \{-1, +1\}$ \emph{a colouring}. The \emph{discrepancy} of a set $S_i$ is defined as $\chi(S_i) = \sum_{u \in S_i} \chi(u)$, and the \emph{discrepancy} of $\mathcal{F}$ is defined as $\max_{i\in [z]}|\chi(S_i)|$. 
\end{definition}

In~\cite[Section 1.1]{Chazelle01}, Chazelle presented a construction of a colouring of small discrepancy assuming infinite precision of computation. Our deterministic algorithm will follow the outline of this construction (although crucial modifications are required in order to overcome the infinite precision assumption), so we quickly restate Chazelle's construction below. The main idea is to assign colours so as to minimize the value of an objective function $G = G(\chi, \{S_i\})$ defined as follows: let $\eps$ be chosen so that $\log \frac{1+\eps}{1-\eps} = \alpha\cdot \sqrt{\log (3z) / k}$ for some constant $\alpha > 2$, and let $p_{i}$ (respectively, $n_{i}$) be the number of $u \in S_i$ such that $\chi(u) = +1$ (respectively, $\chi(u) = -1$) for $i \in [z]$. Define
\[ G_i = (1+\eps)^{p_{i}} (1-\eps)^{n_{i}} + (1+\eps)^{n_{i}} (1-\eps)^{p_{i}} \mbox{\quad and \quad}G = \sum_{i \in [z]} G_i\]
Chazelle's construction assigns colours to one element of $U$ at a time, without ever backtracking. To assign a colour to an element $u$, it performs the following three simple steps. First, it computes $G^+$, the value of $G$ assuming $\chi(u) = +1$. Second, it computes $G^-$, the value of $G$ assuming $\chi(u) = -1$. Finally, if $G^+ \le G^{-}$, it sets $\chi(u) = +1$ and $G = G^{+}$, and otherwise it sets $\chi(u) = -1$ and $G = G^{-}$. Note that for each $i \in [z]$, we have
\begin{equation*}\label{eq:G_decreases}
\begin{split}
(1+\eps)^{p_{i}+1} (1-\eps)^{n_{i}} + (1+\eps)^{n_{i}} (1-\eps)^{p_{i}+1} + (1+\eps)^{p_{i}} (1-\eps)^{n_{i}+1} + (1+\eps)^{n_{i}+1} (1-\eps)^{p_{i}} \\
 = 2 \cdot \left( (1+\eps)^{p_{i}} (1-\eps)^{n_{i}} + (1+\eps)^{n_{i}} (1-\eps)^{p_{i}}\right)
\end{split}
\end{equation*}
and therefore the value of $G$ can only decrease. This implies an important property of Chazelle's construction: since at initialization we have $n_i = p_i = 0$ for all $i \in [z]$ and therefore $G = 2z$, we have $G_i \le G \le 2z$ for $i \in [z]$ at any moment of the construction. Let us show that small values of $G_i$'s imply small discrepancy. In order to do this, we follow the outline of~\cite{Chazelle01}, but use a slightly higher bound for $G_i$'s to be able to apply this lemma later.

\begin{lemma}[\cite{Chazelle01}]\label{lm:discrepancy-inf}
If after all elements of $U$ have been assigned a colour we have $G_i \le 3z$ for all $i \in [z]$, then the discrepancy of the resulting colouring is at most $\alpha \cdot \sqrt{k\log (3z)}$ for any constant $\alpha>2$.
\end{lemma}
\begin{proof}
After all elements of $U$ have been assigned a colour, we have
\[ G_i = (1-\eps^2)^{\frac{|S_i|-\chi(S_i)}{2}} \left[ (1+\eps)^{\chi(S_i)} + (1-\eps)^{\chi(S_i)} \right].\]
Consequently,
\[ (1-\eps^2)^{\frac{|S_i|-\chi(S_i)}{2}} (1+\eps)^{\chi(S_i)} \le 3z.\]
By taking the logarithm of both sides, we obtain
\[|S_i| \log(1-\eps^2) + \chi(S_i) \log\frac{1+\eps}{1-\eps} \le 2 \log (3z).\]
For all $0 < \eps < 1$,  
\[\log (1-\eps^2) \ge - \left[ \frac{1}{2} \log \frac{1+\eps}{1-\eps} \right]^2\]
which implies that
\[ -|S_i| \left[ \frac{1}{2} \log \frac{1+\eps}{1-\eps} \right]^2 + \chi(S_i) \log\frac{1+\eps}{1-\eps} \le 2 \log (3z).\]
Substituting $\log \frac{1+\eps}{1-\eps} = \alpha\cdot \sqrt{\log(3z) / k}$, we finally obtain for any $\alpha>2$:
\begin{equation*}
\chi(S_i) \le  2\log (3z) / (\alpha \cdot \sqrt{\log(3z) / k}) + (\alpha/4) \cdot |S_i|\sqrt{\log(3z) / k} \le \alpha \cdot \sqrt{k \log (3z)}\qedhere
\end{equation*}
\end{proof}

We will show a deterministic algorithm that computes a colouring for which the values $G_i$ are bounded by $3z$. By Lemma~\ref{lm:discrepancy-inf}, this implies that the discrepancy is bounded by $\alpha \cdot \sqrt{k\log (3z)}$. We must overcome several crucial issues: first, we must explain how to compute $\eps$. Second, we must design an algorithm that uses only multiplications and additions so as to be able to control the accumulated precision error. And finally, we must explain how to remove the assumption of infinite precision and to ensure that we never operate on numbers that are too small.

\begin{proposition}\label{prop:eps}
Assume $k > \log (3z)$. There is a deterministic algorithm that computes $\eps\in (0,1)$ such that $\log \frac{1+\eps}{1-\eps} = \alpha\cdot \sqrt{\log(3z)/k}$ for some constant $\alpha > 2$ in $\Oh(\log (zk))$ time. Both $\eps$ and $1-\eps$ are bounded from below by $1/(kz)^{\Oh(1)}$.
\end{proposition}
\begin{proof}
We present the algorithm as a sequence of four steps. Let $\alpha_1 = \lceil \log(3z)\rceil /k$, where $\lceil \log (3z) \rceil$ is computed in $\Oh(\log z)$ time by incrementing a counter $t_1$ until $2^{t_1}\geq 3z$. As $k > \log (3z)$, it holds that $\alpha_1 < 2$. Compute $\alpha_2$ such that $\sqrt{\alpha_1} \le \alpha_2 \le \sqrt{2 \alpha_1}$ by incrementing a counter $t_2$ until $\alpha_1 \le 2^{-t_2} \le 2\alpha_1$ and returning $2^{-t_2/2}$. We have $\alpha_2 < 2$. This step takes $\Oh(\log k)$ time. For $0 < x < 64$, we have $\log(1+x) > x / 8\sqrt 2$. For $\alpha_3 = 16 \sqrt 2 \cdot \alpha_2$, we have $\alpha_3 < 64$ and $\log (1+\alpha_3) > 2 \alpha_2$. Finally, we have $\log \frac{1+\eps}{1-\eps} = \log \left(1+\frac{2}{1/\eps-1}\right)$. It follows that we can take $\frac{2}{1/\eps-1} = \alpha_3$, or equivalently, $\eps = 1-\frac{2}{2+\alpha_3}$, which concludes the proof.

Note that both $\eps$ and $1-\eps$ are bounded from below by $1/(kz)^{\Oh(1)}$:
\[\eps  = 1-\frac{2}{2+\alpha_3} = 1 - \frac{1}{1+8\sqrt{2} \alpha_2} \ge 1 - \frac{1}{1+8\sqrt{2}\sqrt{\alpha_1}} \ge 1 - \frac{\sqrt{k}}{\sqrt{k}+8\sqrt{2}} = \frac{8\sqrt{2}}{\sqrt{k}+8\sqrt{2}}\]
\[1-\eps  = \frac{2}{2+\alpha_3} > \frac{2}{2+64} = 1/33\qedhere\]
\end{proof}

\noindent
We can implement Chazelle's construction to use only multiplications and additions via segment trees.

\begin{proposition}\label{prop:add_multiply}
Assume that $(1+\eps)$ and $(1-\eps)$ are known. Chazelle's construction can be implemented via $\Oh(zk\log z)$ addition and multiplication operations.
\end{proposition}
\begin{proof}
We maintain a complete binary tree on top of $\{1,2,\ldots,2^t\}$, where $2^{t-1} < z \le 2^t$. At any moment, the $(2i-1)$-th leaf stores $(1+\eps)^{p_{i}} (1-\eps)^{n_{i}}$ and the $(2i)$-th leaf stores $(1+\eps)^{n_{i}} (1-\eps)^{p_{i}}$ for all $i \in [z]$, while all the other leaves store value $0$. Each internal node stores the sum of the values in the leaves of its subtree. In particular, the root stores the value $G$. To update $G$ after setting $\chi(u)$ for $u \in U$, we must update the values stored in the $(2i-1)$-th and $(2i)$-th leaves for all $i$ such that $u \in S_i$, as well as the sums in the $\Oh(\log z)$ internal nodes above these leaves. For each leaf, we use one multiplication operation (we must multiply the value by $(1+\eps)$ or $(1-\eps)$ as appropriate), and for each internal node we use one addition operation. In total, we need $\Oh(\sum_{i \in [z]} |S_i| \log z)=\Oh(zk\log z)$ addition and multiplication operations.
\end{proof}

We are now ready to remove the infinite precision assumption and to show the final result of this section. Our algorithm will follow the outline of Proposition~\ref{prop:add_multiply}, but the addition and the multiplication operations will be implemented with precision $\Delta$. Moreover, we will guarantee that the algorithm only works with values in $[\Delta, \Oh(z)]$, which will imply that both arithmetic operations can be performed in constant time and that the algorithm takes $\Oh(zk\log z)$ time. 

\begin{theorem}\label{th:discrepancy}
\assumptions, one can find deterministically in $\Oh(zk\log z)$ time a colouring $\chi : U \rightarrow \{-1, +1\}$ such that $\max_{i \in [z]} |\chi(S_i)| \le \alpha \cdot \sqrt{k\log (3z)}$ for some constant $\alpha > 2$.
\end{theorem}
\begin{proof}
 Let $n = |U|$ and $|S_i| \leq k \leq n$. If $k \le \log (3z)$, then for any colouring $\max_{i \in [z]} |\chi(S_i)| \leq k \le \sqrt{k\log (3z)}$. From now on, we assume $k > \log (3z)$. 
 
We first compute $\eps$ as explained in Proposition~\ref{prop:eps}. After having computed $\eps$ and $1-\eps$, the algorithm initializes a complete binary tree on top of $\{1,2,\ldots,2^t\}$, where $2^{t-1} < z \le 2^t$. The algorithm assigns $1$ to every leaf $i \in [2z]$, and $0$ to all other leaves, and then performs a bottom-up traversal to compute the values of inner nodes as the sum of their children.

Then we proceed as in Proposition~\ref{prop:add_multiply}, that is to update $G$ after setting $\chi(u)$ for $u \in U$, we update the values stored in the $(2i-1)$-th and $2i$-th leaves for all $i$ such that $u \in S_i$, as well as the sums in the $\Oh(\log z)$ internal nodes above these leaves. We implement addition and multiplication with precision $\Delta$, which means that instead of the true value $y$ the algorithm obtains a value $y'$ such that $(1-\Delta) y \le y' \le (1+\Delta) y$. Moreover, to ensure that the values the algorithm operates with are never too small, we apply the following workaround. For each leaf $i$, we store a counter $r_i$ denoting that the value in the leaf $i$ should have been multiplied by $(1-\eps)^{r_i}$, but was not in order not to store numbers below~$\Delta$. Whenever the value in the leaf is multiplied by $(1+\eps)$, we also multiply it by $(1-\eps)^{r'}$, where $r' \le r_i$ is the largest integer such that the value is still larger than $\Delta$ and  update $r_i = r_i - r'$. Then we update the values in the inner nodes on the path from the leaf $i$ to the root and, while summing the values of children, we treat the values in the leaves that have $r_i > 0$ (in other words, the values that are smaller than~$\Delta$) as zeros.

Recall that after assigning a colour to an element $u \in U$, we must update the value~$G$. We claim that at any moment, the absolute difference (``the absolute error'') between the value~$G$ computed by our algorithm and the value $G$ computed by the algorithm of Proposition~\ref{prop:add_multiply} with infinite precision, is $\Oh(\Delta \cdot n z^2 k \log z)$. Below we call the latter value ``the true value'' of~$G$. It follows that we can choose $\Delta$ small enough so that after our algorithm has assigned colours to all elements of $U$, the value of $G$ will be bounded by $3z$. By Lemma~\ref{lm:discrepancy-inf}, this implies that the discrepancy of the constructed colouring is bounded by $\le \alpha \cdot \sqrt{k\log (3z)}$, where the constant $\alpha$ is as in Proposition~\ref{prop:eps}. 

We show the claim by induction. Namely, we show that after we have assigned colours to~$j$ elements, the absolute error is $\Oh(\Delta \cdot j z^2 k \log z)$. For $j = 0$, the claim obviously holds. Consider now $j > 0$. The value $G$ computed by the algorithm can be different from the true value of $G$ at this step for three reasons: 
\begin{enumerate}
\item The values in some leaves are replaced with zeros.
\item Addition and multiplication are implemented with precision $\Delta$.
\item We decide the colouring based on approximate values of $G^+, G^-$.
\end{enumerate}
Now we bound the absolute error between the value of $G$ of the solution computed by our algorithm and the true value of $G$ at this step.
Recall that the total number of arithmetic operations in the algorithm is $\Oh(zk\log z)$. It follows that $G^+, G^-$ are the values of arithmetic expressions with $\Oh(zk\log z)$ addition and multiplication operations. Underestimating the values of leaves, we additionally decrease the values $G^{+}, G^{-}$ by at most $2 \Delta z$. Implementing addition and multiplication with precision $\Delta$, we compute the sum of the remaining terms in the arithmetic expressions with precision $\Oh(\Delta \cdot zk \log z)$. 
By the definition of $G^{+}, G^{-}$ and the induction assumption, the true values of $G^{+}, G^{-}$ at this step are bounded from above by $6z$. It therefore follows that at step $j$ we add at most $\Oh(\Delta \cdot z^2 k \log z)$ to the absolute error. By the induction assumption, the total absolute error at step $j$ is 
\[\Oh(\Delta \cdot z^2 k \log z) + \Oh(\Delta \cdot (j-1) z^2 k \log z) = \Oh(\Delta \cdot j z^2 k \log z).\]
This implies that the value of $G$ can be bounded by $3z$ for $\Delta$ small enough, which concludes the proof.
\end{proof}

Theorem~\ref{th:discrepancy} can be used to partition the universe $U$ into a small number of subsets such that the intersection of every subset of the partition and every set $S_i$ is small.  We start with a simple technical lemma.

\begin{lemma}\label{lem:iterations}
Consider a process that starts with $x_{0} = x$, and keeps computing $x_{i+1} := \lfloor x_{i}(1/2+1/\sqrt{x_{i}}) \rfloor $ as long as $x_{i} > 4$. The process ends after at most $\log x+\Oh(\log^{*}x)$ steps. 
\end{lemma}
\begin{proof}
We claim that after at most $t=\log x - 2\log \log x+8$ steps of the process we have that $x_{t} \leq \log^{2} x$. Assume otherwise, that is, after $t$ steps we still have $x_{t} > \log^{2} x$. But then, for each $i=0,1,\ldots,t-1$, $x_{i+1} = \lfloor x_{i}(1/2+1/\sqrt{x_{i}})\rfloor \leq x_{i}(1/2+1/\log x)$, and therefore

\[x_{t} \leq x/2^{t} (1+2/\log x)^{t} \leq x/2^{t} e^{2t/\log x}\]
Substituting $t$ we obtain
\[x_{t} \leq x/2^{\log x-2\log\log x+8} e^{2+(16-4\log\log x)/\log x} \leq \log^{2}x \cdot e^{2+12/\log x} / 2^8 \leq \log^{2} x\]
where the last inequality holds for $x\geq11$, which leads to contradiction.
For sufficiently large $x\geq c$ we have $\log^{2}(\log^{2}x) \leq \log x$. Thus, by repeating the above reasoning $\Oh(\log^{*}x)$ times we obtain that after $\log x+\Oh(\log^{*}x)$ steps the value of $x$ has decreased to at most $\max\{11,c\}$. Then, using the fact that $x_{i+1} < x_{i}$ as long as $x_{i} > 4$, we conclude that after additional $\Oh(1)$ iterations the value of $x$ decreases to at most $4$, and so the process terminates. 
\end{proof}

\begin{lemma}\label{lm:main_hash_function}
 \assumptions, one can construct deterministically in $\Oh(|U| \log z \log k)$ time a function $f: U \rightarrow [k \cdot 2^{\Oh(\log^* k)}]$ such that for each $c \in [k \cdot 2^{\Oh(\log^* k)}]$ and for each $S_i$, the intersection of $\{u \in U \,|\, f(u) = c\}$ and $S_i$ contains $\Oh(\log z)$ elements.
\end{lemma}
\begin{proof}
We can reformulate the statement of the lemma as follows. We must show that there is a partitioning of $U$ into subsets $X_c = \{u \in U: f(u) = c\}$ such that for every $S_i$, the intersection $X_c \cap S_{i}$ has size at most $\Oh(\log z)$.

We partition $U$ recursively using the procedure from Theorem~\ref{th:discrepancy}. We start with a single set $X = U$. Suppose that after several steps we have a partitioning of $U$ into sets $X_c$ such that $|S_{i} \cap X_c| \le y$ for all $i$ and $c$ and some integer $y$. We then apply Theorem~\ref{th:discrepancy} to the sets $X_c$. Using the colouring output by the lemma, we partition each set $X_c$ into sets $X_{c_0}$ and $X_{c_1}$, where the former contains all the elements of $X_c$ of colour $-1$ and the latter all the elements of $X_c$ of colour $+1$. For $j\in\{0,1\}$ we choose $c_j$ (and also the value of $f(x)$ for $x\in X_{c_j}$) so that its binary representation equals the binary representation of $c$ appended with $j$. By Theorem~\ref{th:discrepancy}, there is a constant $\alpha$ such that $|S_i \cap X_{c_0}|, |S_{i} \cap X_{c_1}| \leq y/2+\frac12 \alpha \cdot \sqrt{y \log (3z)} \leq y(1/2+1/\sqrt{y/\alpha^2\log (3z)})$. We continue this process until $|S_{i} \cap X_c| \le 4 \alpha^2 \log (3z)$ for all $i$ and $c$. 

It remains to bound the number of iterations. By setting $x = k / \alpha^2 \log (3z)$ in Lemma~\ref{lem:iterations}, we obtain that we need at most $\log x + \Oh(\log^* x) \le \log k + \Oh(1) + \Oh(\log^{*} k) = \log k + \Oh(\log^* k) = t$ recursive applications of the partition procedure implemented with Theorem~\ref{th:discrepancy} to ensure that every set $S_i$ has at most $4 \alpha^2 \log (3z) = \Oh(\log z)$ elements in common with every $X_{c}$. Therefore, the size of the image of $f$ is bounded by $2^t = k 2^{\Oh(\log^{*} k)}$. The overall construction time is $\Oh(|U| \log z \log k)$.
\end{proof}

\subsection{Superimposed Codes}
We are now ready to show an efficient construction algorithm for data-dependent superimposed codes (see Definition~\ref{def:codes}). At a high level, we will construct a family of functions which, combined with the partition~$f$ from Lemma~\ref{lm:main_hash_function}, will give us the superimposed code.

\newcommand{\SIwpar}[1]{\eps^{-1}\log^2 #1}
\newcommand{\SIellshortpar}[2]{\epsdep #1\log^5 #2}
\newcommand{\ttime}[1]{\Oh(\eps^{-1} #1\log^2 #1)}
\newcommand{\famH}{\mathcal{H}(U,d)}
\newcommand{\numb}{\textsc{num}}
\newcommand{\pol}{\textsc{pol}}
\newcommand{\Deg}{\texttt{deg}}
\newcommand{\Mod}{\texttt{Mod}}
\newcommand{\ZZ}{\mathds{Z}}

\begin{theorem}\label{th:superimposed_codes}
\assumptions, one can construct an $(\{S_i\}, (1-\eps) w)$-superimposed code of weight $w = \Oh(\eps^{-1}\log^2|U|)$ and $\ell  = \Oh(\SIellshortpar{k}{|U|})$ in $\ttime{|U|}$ time and space. 
\end{theorem}
\begin{proof}
By applying Lemma~\ref{lm:main_hash_function}, we obtain in $\Oh(|U| \log z \log k)=\Oh(|U|\log^2|U|)$ time a function $f: U\rightarrow [{k \cdot 2^{\Oh(\log^{*}k)}}]$ which gives a partitioning of $U$ into subsets $X_{c} = \{u \in U \,|\, f(u) = c\}$, such that for some constant~$\alpha$, for every $c$ and $i$ holds $|X_{c} \cap S_i| \leq \alpha \log z$.

Consider the ring of polynomials $\ZZ_2[x]$. Let $U = \{u_1, u_2, \ldots, u_{zk}\}$. We define a mapping $\pol : U \rightarrow \ZZ_2[x]$ as follows. Let $u = u_q$ and $q = \overline{q_t q_{t-1} \ldots q_0}$ be the binary representation of~$q$, where $t = \lfloor{\log |U|\rfloor}$, then $\pol(u) = \sum_{i = 0}^{t} q_i x^i$. 

Let $\famH$ be the family of functions $h_p: U\rightarrow \mathbb{F}_{2^d}$ of the form $h_p(u)= (\pol(u) \bmod p)$ for all irreducible polynomials $p$ of degree $d$. By Gauss's formula \cite{ChebouluSM10,Gauss1889}, there are $\Theta(2^d/d)$ irreducible polynomials of degree $d$ over  $\ZZ_2$, and so is the size of the family $\famH$.
Consider two distinct polynomials $x,y$ of degree $t$.
Observe that there are at most $t / d$ irreducible polynomials $p$ that hash both $x$ and $y$ to the same value $h_p(x)=h_p(y)$, because $\ZZ_2[x]$ is a unique factorization domain \cite{Gauss1889}. We choose $d$ in such a way that the probability that $x,y$ are hashed to the same value while choosing a hash function uniformly at random from $\famH$ is bounded by $\eps/(\alpha \log z)$: $\frac{t/d}{\Theta(2^d/d)} \leq \frac {\eps}{\alpha\log z}$ and hence we can choose $d = \Theta(\log \frac{t\log z}{\eps})$.

If $d > t$, then $\eps<\frac{\log^2|U|}{|U|}$ and we can take $\ell=|U|, w=1$ and set $C_{u_q}=\{q\}$. From now on, assume $d\leq t$. Let $f$ be as in Lemma~\ref{lm:main_hash_function}. 
Consider $u\in U$ such that $u \in X_c$, where $c=f(u)\in[k\cdot 2^{\Oh(\log^*k)}]$. We define $C_u$ as follows:
\[C_u = \{H_p(u) = \numb(h_p(u)) + 2^d\cdot \numb(p) + 4^d\cdot c \; | \; h_p\in\famH\},\]
where the mapping $\numb(q)$ treats a polynomial $q = \sum_{i = 0}^{d-1} q_i x^i$ as a $d$-bit number $\overline{q_{d-1} \ldots q_0}$. Clearly, $w=|C_u|=\Oh(2^d/d)=\Oh(2^d)=\Oh(\frac{t\log z}{\eps})=\Oh(\eps^{-1}\log^2|U|)$ and $C_u\subseteq [l]$ where:
\[\ell = 2^d \cdot 2^d \cdot k 2^{\Oh(\log^{*} k)} = \frac{t^2\log^2z}{\eps^2}\cdot k \cdot 2^{\Oh(\log^{*} k)} =  \Oh(\eps^{-2} k\log^5 |U|).\]

We claim that the obtained code is a $(\{S_i\},(1-\eps) w)$-superimposed code. Consider any $S_i$ and $u \notin S_i$. We need to count elements of $C_u$ that do not belong to any $C_v$, for $v \in S_i$. Let $c=f(u) \in [k\cdot 2^{\Oh(\log^{*} k)}]$ and so $u \in X_{c}$. By construction, $|X_{c} \cap S_i| \leq \alpha \log z$. Thus, by the union bound, the probability that $h_p(u)=h_p(x)$ for some $x \in X_{c} \cap S_i$ is at most~$\eps$ for $h_p$ chosen uniformly at random from $\famH$. Recall that $C_u$ consists of elements $H_p(u)=\numb(h_p(u)) + 2^d\cdot \numb(p) + 4^d\cdot c$ for $h_p \in \famH$. The number of irreducible polynomials $p$ such that $H_p(u)=H_p(x)$ for some $x \in X_{c} \cap S_i$
is at most $\eps \cdot w$. Consequently, at least $w-\eps \cdot w=(1-\eps) w$
elements of~$C_u$ do not belong to any $C_v$, for $v \in S_i$.

We now show that we can construct the above superimposed codes in $\Oh(|U|w)$ time.
To this end, we need to generate all irreducible polynomials of degree $d$ and to explain how we compute remainders modulo these polynomials. Note first that as we only operate on polynomials of degree $\leq t = \Oh(\log |U|)$, they fit in a machine word and hence we can subtract two polynomials or multiply a polynomial by any power of $x$ in constant time. We can now use this to generate the irreducible polynomials and compute the sets $C_u$ at the same time. We maintain a bit vector $I$ that for each polynomial $p$ of degree $\le d$ stores an indicator bit equal to $1$ iff $p$, i.e. iff its remainder modulo any polynomial of degree smaller than $\deg(p)$ is not zero. We consider the polynomials of degree $0, 1, 2, \ldots, d$ in order. For every irreducible polynomial $p$, we compute a table $\Mod_p[q] = (q \bmod p)$ for all polynomials $q$ of degree $\leq t$ in overall $\Oh(|U|)$ time using dynamic programming with the following recursive formula:
\[
    \Mod_p[q]=
    \begin{cases}
      q, & \text{if } \Deg(q) < \Deg(p)\\
      \Mod_p[q - p\cdot x^{\Deg(q) -\Deg(p)}], & \text{otherwise}
    \end{cases}
\]
We use the table to compute $H_p(u)$ for all $u\in U$. Also, if for a polynomial $q$ the remainder is zero, we zero out the corresponding bit in $I$. Here we use the fact that $d \le t$ to guarantee that we will find all irreducible polynomials of degree $\le d$ in this way.

As there are $w$ irreducible polynomials, in total we spend $\Oh(|U|w)=\Oh(\eps^{-1}|U|\log^2|U|)$ time. At any moment, we use $\Oh(|U|)$ space to store the table and $\Oh(\eps^{-1} |U| \log^2|U|)$ space to store the codes. 
\end{proof}

\section{Upper Bounds for Generalised Pattern Matching}
\label{sec:algorithms}
In this section, we present new algorithms for the parameters $\D$, $\S$, and $\I$. Our algorithms for the parameters $\D$ and $\S$ share similar ideas, so we present them together in Section~\ref{sec:d_and_s}. The algorithm for $\I$ is presented in Section~\ref{sec:ranges}.

We start by recalling the formal statement of the \DontCares problem that will be used throughout this section.

{\defproblem{\textsc{\DontCares} (counting, binary alphabet)}{A text $T\in \{0, 1, ?\}^n$ and a pattern $P\in \{0, 1, ?\}^m$, where ``?'' is a \emph{don't care character} that matches any character of the alphabet.}{For each $i \in [n-m+1]$, the number of positions $j \in [m]$ such that $T[i+j-1]$ does not match $P[j]$.}

\noindent Clifford and Clifford~\cite{CLIFFORD200753} showed that this problem can be solved in $\Oh(n\log m)$ time.

\subsection{Parameters \texorpdfstring{$\D$}{D} and \texorpdfstring{$\S$}{S}}
\label{sec:d_and_s}
We first show Monte Carlo algorithms for the reporting and counting variants of \GPM, and then de-randomise them using the data-dependent superimposed codes of Section~\ref{sec:superimposed_codes}.

\subsubsection{Randomised Algorithms}
We start by presenting a new reporting algorithm for the parameter $\D$. It does not improve over the algorithm of~\cite{Muthukrishnan95}, but encapsulates a novel idea that will be used by all our algorithms for the parameters $\D$ and $\S$. Essentially, we use hashing to reduce $\Sigma_T$ to a smaller set of characters of size $p=\Theta(\D)$ while preserving occurrences of the pattern in the text with constant probability, and then show that this smaller instance of \GPM can be reduced to $p=\Theta(\D)$ instances of \DontCares.

\begin{theorem}\label{th:UB_d_rand}
Let $\D$ be the maximum degree in the matching graph $M$ and $c$ be any constant fixed in advance. There is a Monte Carlo algorithm that solves the reporting variant of \GPM in $\Oh(\D n \log m \log n)$ time. The error is one-sided (only false positives are allowed), and the error probability is at most $1/n^c$.
\end{theorem}
\begin{proof}
If $\D > m$, we can use a naive algorithm that compares the pattern and each $m$-length substring of the text character-by-character and uses $\Oh(mn) = \Oh(\D n)$ time in total. Below we assume $\D \le m$. We can also assume $|\Sigma_T| \le n$.

We first choose a 2-wise independent hash function $h: \Sigma_T \rightarrow [2\D]$ of the form $h(x) = ((a \cdot x + b) \mod p) \mod (2\D)+1$, where $p \ge |\Sigma_T|$ is a prime, and $a, b$ are chosen independently and uniformly from $\mathbb{F}_p$. Note that we can find a prime $p$ such that $n \le p \le 2n$, in $\Oh(n)$ time. Consider a matching graph $M'$ on the set of vertices $[p] \cup \Sigma_P$. For every character $b = P[j]$ and for every character $a \in \Sigma_T$ in the adjacency list of $b$, we add an edge $(h(a), b)$ to $M'$. Overall, it takes $\Oh(\D m) = \Oh(\D n)$ time.

We claim that if $M$ does not contain an edge $(a, b)$, then the probability of $M'$ to contain an edge $(h(a),b)$ is at most $1/2$. By definition, if $(h(a),b)$ belongs to $M'$, then there exists a character $a' \in\Sigma_{T}$ such that $(a' , b)$ is in $M$ and $h(a') = h(a)$. Since $h$ is 2-wise independent, for a fixed character $a' $ the probability of $h(a') = h(a)$ is $1/(2\D)$. Because the degree of $b$ is at most~$\D$, the probability of such event is at most $1/2$ by the union bound. 

Consider a text $T'$, where $T'[i] = h(T[i])$. If $T[i, i+m-1]$ does not match $P$ under $M$, then $T'[i,i+m-1]$ does not match $P$ under $M'$ with probability $\ge 1/2$. Indeed, suppose that for some $j \in [m]$, $T[i+j-1]$ and $P[j]$ do not match under $M$, or equivalently, an edge $(T[i+j-1], P[j])$ does not belong to $M$. From above, with probability at least $1/2$, $h(T[i+j-1])$ and $P[j]$ do not match under $M'$. It follows that we can use the \GPM algorithm for $M'$, $T'$, and $P$ to eliminate every non-occurrence of $P$ in~$T$ with probability at least $1/2$. We can amplify the probability by independently repeating the algorithm $c\log n$ times.

It remains to explain how to solve \GPM for $M'$, $T'$, and~$P$. We use the fact that the size of the alphabet of $T'$ is $\Oh(\D)$. For every $a\in [2\D]$ we create a new text $T'_{a}[1, n]$ and a new pattern $P_{a}[1, m]$ as follows:

\begin{tabular}{c @{\hspace{2cm}} c}
$T'_{a}[j]=
\begin{cases}
\texttt{0} & \mbox{ if } T'[j] = a,\\
\texttt{?} & \mbox{ otherwise.}
\end{cases}
$
&
$P_{a}[j]=
\begin{cases}
\texttt{0} & \mbox{ if } a \mbox{ matches } P[j] \mbox{ under } M',\\
\texttt{1} & \mbox{ otherwise}.
\end{cases}
$
\\
\end{tabular}

We can construct $T'_a$ and $P_a$ in $\Oh(n+m) = \Oh(n)$ time, or in $\Oh(\D n)$ total time for all $a \in [2\D]$. It is not hard to see that $T'[i,i+m-1]$ matches $P$ if and only if $T'_a[i,i+m-1]$ matches~$P_a$ for all $a \in [2\D]$. Therefore, to solve \GPM for $M'$, $T'$, and $P$, it suffices to solve the $2\D$ instances of \DontCares. By~\cite{CLIFFORD200753}, this can be done in total $\Oh(\D n \log m)$ time. As we repeat the algorithm $c\log n$ times, the theorem follows.
\end{proof}

We now show a new randomised algorithm for the parameter $\S$. At a high level, we divide~$\Sigma_P$ into heavy and light characters based on their degree in $M$ (a character of $\Sigma_{P}$ is called heavy when it matches many characters of $\Sigma_{T}$, and light otherwise). The number of heavy characters is relatively small, and we can eliminate all substrings of $T$ that do not match $P$ because of heavy characters by running an instance of \DontCares for each of them. For light characters, we apply Theorem~\ref{th:UB_d_rand}.

\begin{theorem}\label{th:UB_s_rand}
Let $\S$ be the number of edges in the matching graph $M$ and $c$ be any constant fixed in advance. There is a Monte Carlo algorithm that solves the reporting variant of \GPM in $\Oh(\sqrt{\S} n\log m\sqrt{\log n})$ time. The error is one-sided (only false positive are allowed), and the error probability is at most $1/n^c$. 
\end{theorem}
\begin{proof}
If $\sqrt{\S} > m$, we can use a naive algorithm that compares each $m$-length substring of the text and the pattern character-by-character and uses $\Oh(mn) = \Oh(\sqrt{\S} n)$ time in total. Below we assume $\sqrt{\S} \le m$. 

We call a character $b\in \Sigma_{P}$ \emph{heavy} when it matches at least $\sqrt{\S/\log n}$ characters $a\in \Sigma_{T}$, and otherwise we call $b$ \emph{light}. Observe that the number of heavy characters is at most $\sqrt{\S\log n}$. 
For every heavy character $b\in \Sigma_{P}$ we solve a separate instance of \DontCares to rule out all substrings of the text that do not match the pattern due to $b$. The instance is formed by creating a text $T_{b}[1,n]$ and a pattern $P_b[1,m]$:

\begin{tabular}{c @{\hspace{2cm}} c}
$T_{b}[j]=
\begin{cases}
\texttt{0} & \mbox{ if } T[j] \mbox{ and } b \mbox{ do not match under } M,\\
\texttt{?} & \mbox{ otherwise.}
\end{cases}
$
&
$
P_{b}[j]=
\begin{cases}
\texttt{1} & \mbox{ if } P[j]=b,\\
\texttt{?} & \mbox{ otherwise.}
\end{cases}
$
\end{tabular}

\noindent
By~\cite{CLIFFORD200753}, we can solve all the instances deterministically in $\Oh(\sqrt{\S} n \log m\sqrt{\log n})$ time. Clearly, $T[i, i+m-1]$ and $P$ do not match because of a heavy character $b$ if and only if $T_b[i, i+m-1]$ and $P_b$ do not match. 

It remains to rule out all substrings $T[i,i+m-1]$ that do not match due to a light character $P[j]$ aligned with non-matching character $T[i+j-1]$. To this end, first replace every heavy character $P[j]$ with the don't care character and then apply the technique of Theorem~\ref{th:UB_d_rand} with $\D = \sqrt{\S/\log n}$. This step takes $\Oh(\sqrt{\S} n \log m \sqrt{\log n})$ time and has one-sided error probability~$1/n^c$.    
\end{proof}

Combining the techniques of Theorems~\ref{th:UB_d_rand},~\ref{th:UB_s_rand} and the approach of Kopelowitz and Porat~\cite{KopelowitzP18}, we obtain the following corollary. 

\begin{corollary}\label{cor:UB_s_d_rand_counting}
Let $c$ be any constant fixed in advance, $\D$ be the maximum degree and $\S$ be the number of edges in the matching graph $M$. There is a $(1-\eps)$-approximation Monte Carlo algorithm that solves the counting variant of \GPM in $\Oh(\min\{\eps^{-1} \D \log n, \sqrt{\eps^{-1} \S \log n}\} \cdot n \log m)$ time. The error probability is at most~$1/n^c$.
\end{corollary}
\begin{proof}
We first explain how to modify the algorithm of Theorem~\ref{th:UB_d_rand} to obtain a counting algorithm with complexity $\Oh(\eps^{-1}\D n \log n\log m)$. As before, we assume $\D /\eps \le m$, otherwise we can use the naive algorithm. We follow the approach of Kopelowitz and Porat~\cite{KopelowitzP18}. Namely, we repeat the following process $\log n$ times. Instead of the hash function used by the algorithm, we choose a $2$-wise independent hash function $h : \Sigma_{T} \rightarrow [q]$, where $2\D / \eps \le q \le 4\D / \eps$ is a prime, and proceed as in Theorem~\ref{th:UB_d_rand} to obtain a new matching graph $M'$ and a new text~$T'$. For each $a \in [q]$, we solve \DontCares for $T'_a$ and $P_a$, and as a result obtain the total number $h'[i]$ of mismatches between $T'_a[i,i+m-1]$ and $P_a$ for all $i\in [n-m+1]$. Denoting by $h[i]$ the true number of mismatches under $M$, by the reasoning from Theorem~\ref{th:UB_d_rand} we have that $\mathbb{E} [h'[i]] \geq h[i] (1-\eps/2)$. Finally, for every $i\in [n-m+1]$ we return the maximum value of $h'[i]$ obtained in all $c \log n$ iterations. By Markov's inequality, with probability at least $1-1/n^{c}$ this is at least $(1-\eps) h[i]$ (and clearly at most $h[i]$). As a result, we obtain an algorithm with time $\Oh(\eps^{-1} \D  n \log m \log n)$.

We now show how to modify the algorithm of Theorem~\ref{th:UB_s_rand}. For that, we define a character $b\in \Sigma_{P}$ \emph{heavy} if it matches at least $\sqrt{\eps \S/\log n}$ characters in $\Sigma_T$. We then note that the number of mismatches due to heavy letters can be computed exactly by solving $\sqrt{\eps^{-1} \S \log n}$ obtained instances of \DontCares, and for the light characters we use the method explained above.  We therefore obtain an algorithm with time $\Oh(\sqrt{\eps^{-1} \S \log n} \, n \log m)$, and the claim follows.
\end{proof}

\subsubsection{Deterministic Algorithms}
We are now ready to give $(1-\eps)$-approximation deterministic algorithms for the counting variant of \GPM for the parameters $\D$ and $\S$. By taking $\eps = 1/2$, the algorithms for the reporting variant follow immediately. We first remind the definition of superimposed codes, which we will use throughout this section.

\codes*

\begin{theorem}\label{th:UB_d_det}
Let $\D$ be the maximum degree in the matching graph $M$. There exists an ${(1-\eps)}$-approximation deterministic algorithm that solves the counting variant of \GPM in $\timeUBddet$ time.
\end{theorem} 
\begin{proof}
First, note that we can assume $\D \le m$ and $\eps \ge 1/m$. If this is not the case, we can run a naive algorithm that compares each $m$-length substring of the text $T$ and the pattern character-by-character in $\Oh(mn) = \Oh (\D n)$ time. 

For each distinct character $b$ of the pattern $P$, consider a set $S_b$ containing all characters in $\Sigma_T$ that match $b$. By definition, $|S_b| \le \D$. We define the universe $U = (\bigcup_{b\in\Sigma_P} S_b) \cup \{\$\}$, where $\$ \notin \Sigma_T$ is a special character that we will need later, $|U|=\Oh(n)$. We apply Theorem~\ref{th:superimposed_codes} that constructs $(\{S_b\}, (1-\eps) w)$-superimposed code for the universe $U$ and sets $S_b$ in $\ttime{n}$ time, where the weight $w = \Oh(\SIwpar{n})$ and the length $\ell  = \Oh(\SIellshortpar{\D}{n})$. 

We define the code of a character $a \in U$ to be a binary vector of length $\ell$ such that its $j$-th bit equals $1$ if $C_a$ contains $j$, and $0$ otherwise. For a character $a' \in \Sigma_T \setminus U$, we define its code to be equal to the code of $\$$. We define the code of a character $b \in \Sigma_P$ to be a binary vector of length $\ell$ such that its $j$-th bit equals $1$ if $\bigcup_{a \in S_b} C_a$ contains $j$, and $0$ otherwise. Next, we create a text $T'[1,n\ell]$ and a pattern $P'[1,m\ell]$ by replacing the characters in respectively $T$ and $P$ by their codes. To finish this step, we replace each $1$ in $P'$ with the don't care character and run the algorithm of  Clifford and Clifford~\cite{CLIFFORD200753} for $T'$ and $P'$ that takes $\Oh(n\ell \log (m \ell)) = \Oh(\epsdep\D n \log^6 n)$ time (here we use $\eps \ge 1/m$).

Let $h'$ be the number of mismatching characters between $P'$ and $T'[(i-1)\cdot \ell+1, (i+m-1) \cdot \ell]$, and $h$ be the number of mismatches between $P$ and $T[i, i+m-1]$. We claim that $(1-\eps) w h \le h' \le w h$. Indeed, if $P[j]$ matches $T[i+j-1]$, then $C_{T[i+j-1]}$ is a subset of $\bigcup_{a \in S_{P[j]}} C_a$. Therefore, if the code of $T[i+j-1]$ contains $1$ in position $k$, the code of $P[j]$ will have $1$ in position $k$ as well. By replacing all $1$s in $P'$ with the don't care characters, we ensure that the corresponding fragments of $P'$ and $T'$ match. On the other hand,  if $P[j]$ does not match $T[i+j-1]$, then from the definition of the code it follows that the distance between the corresponding chunks of $P'$ and~$T'$ will be at least $(1-\eps) w$ and at most $w$.
\end{proof}

To show a deterministic algorithm for the parameter $\S$, we again consider the partition of the alphabet $\Sigma_P$ into heavy and light characters. To count the mismatches caused by some heavy character, we create an instance of \DontCares. As the number of heavy characters is small, the total number of the created instances is small as well. For light characters, we use the superimposed codes similarly as in Theorem~\ref{th:UB_d_det}.

\begin{theorem}\label{th:UB_s_det}
Let $\S$ be the number of edges in the matching graph $M$. There exists an $(1-\eps)$-approximation deterministic algorithm that solves the counting variant of \GPM in $\timeUBsdet$ time. 
\end{theorem}
\begin{proof}
The structure of the algorithm is similar to that of Theorem~\ref{th:UB_s_rand}.  We define a character $b\in \Sigma_{P}$ \emph{heavy} if it matches at least $\eps \sqrt{\S} / \log^{5/2} n$ characters in $\Sigma_T$. The number of heavy characters is at most $\eps^{-1 }\sqrt{\S}  \log^{5/2} n$. To count the number of mismatches caused by a heavy character $b\in \Sigma_{P}$, we create a text $T_{b}[1,n]$ and a pattern $P_b[1,m]$:

\begin{tabular}{c @{\hspace{2cm}} c}
$T_{b}[j]=
\begin{cases}
\texttt{0} & \mbox{ if } T[i] \mbox{ and } b \mbox{ do not match under } M,\\
\texttt{?} & \mbox{ otherwise.}
\end{cases}
$
&
$
P_{b}[j]=
\begin{cases}
\texttt{1} & \mbox{ if } P[j]=b,\\
\texttt{?} & \mbox{ otherwise.}
\end{cases}
$
\end{tabular}

\noindent and run the algorithm of Clifford and Clifford~\cite{CLIFFORD200753} for $T_b$ and $P_b$. In total, this step takes $\Oh(\eps^{-1 } \sqrt{\S} n \log^{7/2} n)$ time. We now need to explain how we count the mismatches due to the light characters of~$P$. We use an algorithm similar to that of Theorem~\ref{th:UB_d_det} for $\D = \eps \sqrt{\S} / \log^{5/2} n$, except that we replace each heavy character of $P$ with the code $?^l$ (the don't care character repeated $l$ times, where $l$ is the length of the superimposed code). 
\end{proof}

\subsection{Parameter \texorpdfstring{$\I$}{I}}
\label{sec:ranges}
In this section we show a deterministic \GPM algorithm for the parameter $\I$. The algorithm solves the counting variant of the problem exactly, and we can immediately derive an algorithm for the reporting version with the same complexities as a corollary. We will need the following technical lemma.

\begin{lemma}\label{lem:partition}
Let $b$ be a parameter, $S= \{ x_{1},x_{2},\ldots,x_{\ell} \}$ be a sequence of integers, and $s=\sum_{i\in[\ell]}x_{i}$. Then $S$ can be partitioned into $\Oh(s/b+1)$ ranges $S_{1},S_{2},\ldots$ such that, for every $i$, either $S_{i}$ is a singleton or the sum of all elements in $S_{i}$ is at most $b$.
\end{lemma}
\begin{proof}
We greedily partition $S$ from left to right. In every step we consider the remaining suffix of the sequence. If its first element is greater than $b$ then we create a singleton range. Otherwise, we choose the longest prefix of the remaining suffix consisting of numbers summing up to at most $b$ and create the corresponding range. Every range in the partition is either a singleton or consists of integers summing up to at most~$b$. It remains to argue that the number of ranges is small. But the number of singletons is less than $s/b$, and for every non-singleton range $S_{i}$ the sum of numbers in $S_{i}\cup S_{i+1}$ (assuming that $S_{i+1}$ exists) is greater than $b$, so the bound of $\Oh(s/b+1)$ follows.
\end{proof}

\noindent
We are now ready to show the main result of the section. 

\begin{theorem}\label{th:det_i}
For each character $a \in \Sigma_P$ consider a minimal set $I(a)$ of disjoint sorted intervals that contain the characters that match $a$, and define $\I = \sum_{j \in [m]} |I(P[j])|$. There is a deterministic algorithm that solves the counting version of \GPM in $\Oh(n\sqrt{\I \log m}+n\log n)$ time.
\end{theorem}
\begin{proof}
If $\I > m^2$, we can use the naive algorithm that compares each $m$-length substring with the pattern character-by-character and takes $\Oh(mn)$ time in total. 

We first make a pass over $T$ and retrieve the set of distinct characters $a_1, a_2, \ldots, a_l$ of $\Sigma_T$ that occur in it, as well as their frequencies. This can be done in $\Oh(n \log n)$ time using a binary search tree. We partition $a_1, a_2, \ldots,a_l$ into ranges as follows. Let $\countchars(c)$, for $c\in\Sigma_T$, be the frequency (i.e. the number of occurrences) of $c$ in $T$. We apply Lemma~\ref{lem:partition} for $b>1$ that will be specified later and the sequence $\countchars(a_1), \countchars(a_2), \ldots, \countchars(a_l)$ which sums up to $n$. 

Let $\Sigma'_T$ be a new alphabet obtained by creating a character for every range in the partition, where $|\Sigma'_T|=\Oh(n/b+1)$.
For $c \in \Sigma'_T$ we denote by $\range(c)$ the range of~$\Sigma_T$ corresponding to $c$, and for $a \in \Sigma_T$ we denote by $\range^{-1}(a)$ the character of~$\Sigma'_T$ corresponding to the range containing $a$. We create a new text $T'[1, n]$ and pattern $P'[1, m]$ as follows. For every $i \in [n]$, we set $T'[i]=\range^{-1}(T[i])$. For every $j \in [m]$, we set $P'[j] = \{c \in \Sigma'_T \,|\, \range(c) \mbox{ contains a character that matches } P[j]\}$. As the number of the ranges is $\Oh(n/b+1)$, the size of the set $P'[j]$ is $\Oh(n/b+1)$. We represent it as a binary vector of length $\Oh(n/b+1)$. Furthermore, we can construct $T'$ in $\Oh(n)$ time, and $P'$ in $\Oh(\I + m(n/b+1))$ time.

After this initial step the algorithm consists of two phases. First, we solve the \textsc{Subset Pattern Matching} for $T'$ and $P'$ that consists of counting, for every $i \in [n-m+1]$, all positions $j \in [m]$ such that $T'[i+j-1]\notin P'[j]$. To this end, we create an instance of \DontCares for every $c \in \Sigma'_T$, namely, we create a text $T'_{c}[1, n]$ and a pattern $P'_{c}[1, m]$ as follows:

\begin{tabular}{c @{\hspace{2cm}} c}
$T'_{c}[i] =
\begin{cases}
\texttt{0} & \mbox { if } T'[i]=c,\\
\texttt{?} & \mbox {otherwise.}
\end{cases}
$
&
$P'_{c}[j]=
\begin{cases}
\texttt{0} & \mbox { if } c \in P'[j],\\
\texttt{1} & \mbox{ otherwise.}
\end{cases}
$
\end{tabular}

\noindent
We can solve all these instances in $\Oh(|\Sigma'_T| n\log m) = \Oh((n/b+1)n\log m)$  time~\cite{CLIFFORD200753}. Summing up the results, we obtain the result for the subset matching problem.

In the second phase, we slightly adjust the results obtained for \textsc{Subset Pattern Matching} to obtain the results for \GPM.  Consider a substring $T[i,i+m-1]$ that does not match $P$ because of a mismatch in position $j$ of the pattern, i.e. $T[i+j-1]$ does not match $P[j]$. We have two possible cases. The first case is when $T'[i+j-1] \notin P'[j]$. In this case, the mismatch is detected by the \textsc{Subset Pattern Matching} algorithm. The second case is when $T'[i+j-1] \in P'[j]$. Observe that in this case, $\range(T'[i+j-1])$ cannot be a singleton and must contain an endpoint of some interval of characters that match~$P[j]$. 

To detect such mismatches, we run the following algorithm. For each $j \in [m]$, we consider the intervals $I(P[j])$ of the characters that match $P[j]$. For every endpoint $c\in\Sigma_T$ of the intervals in $I(P[j])$, we iterate over all $a \in \range^{-1}(c)$ such that $a$ does not match $P[j]$ and all occurrences of $a$ in the text. 
Summing over all $j$ and $a$, there are in total $\Oh(\I \cdot b)$ of the occurrences due to the properties of the partition and the fact that $\range(T'[i+j-1])$ is not a singleton. We can find the occurrences in $\Oh(\I \cdot b + n + mn/b)$ time as follows. First we find the ranges containing the endpoints in $\Oh(\I + m(n/b+1))$ time similarly to above, and we can generate the lists of occurrences of every character $a \in \Sigma_T$ in $T$ by one pass over $T$ in $\Oh(n \log n)$ time. For each such occurrence $T[k] = a$ that does not match $P[j]$, we increment the number of mismatches for the substring $T[k-j+1,k+m-j]$. This correctly detects every mismatch that has not been accounted for in the first phase, and hence allows counting all mismatches in $\Oh(\I+m(n/b+1)+(n/b+1)n\log m+n\log n+\I \cdot b)=\Oh(n^2\log (m)/b+n\log n+\I \cdot b)$ total time. Substituting $b=n\sqrt{\log m/\I}$ gives us the claim of the theorem.
\end{proof}

\begin{corollary}\label{cor:threshold}
There is a deterministic algorithm that solves the counting variant of the threshold pattern matching problem in $\Oh(n (\sqrt{m \log m} + \log n))$ time.
\end{corollary}

\section{Lower Bounds for \GPM}
\label{sec:lb}
In this section we give lower bounds for \GPM algorithms.
All the lower bounds are presented for the reporting variant of \GPM, so they immediately apply also to the counting variant.
Recall that we assume to have access to three oracles that can answer the following questions about the matching graph $M$ in $\Oh(1)$ time:

\begin{enumerate}
\item Is there an edge between $a\in\Sigma_{T}$ and $b\in\Sigma_{P}$?
\item What is the degree of a character $a\in\Sigma_{T}$ or $b\in\Sigma_{P}$?
\item What is the $k$-th neighbor of $a\in\Sigma_T$?
\end{enumerate}

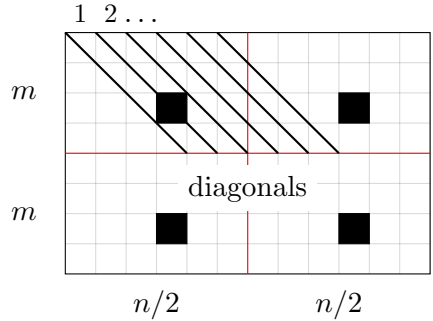
\begin{figure}[h!]
\begin{center}
\begin{tikzpicture}[scale = 0.8]
\foreach \x in {1,2,...,12} {
	\draw[gray, opacity=0.3] (\x*0.5,0)--(\x*0.5,4);
}
\foreach \x in {1,2,...,8} {
	\draw[gray, opacity=0.3] (0,\x*0.5,0)--(6,\x*0.5);
} 

\draw[fill=black] (1.5,2.5) rectangle (2,3);
\draw[fill=black] (4.5,2.5) rectangle (5,3);
\draw[fill=black] (1.5,.5) rectangle (2,1);
\draw[fill=black] (4.5,.5) rectangle (5,1);
 
\draw (0,0) rectangle (6,4);
\draw[red] (0,2)--(6,2);
\node[below] at (1.5,-0.1) {$n/2$};
\node[below] at (4.5,-0.1) {$n/2$};
\draw[red] (3,0)--(3,4);
\node[left] at (-0.3,1) {$m$};
\node[left] at (-0.3,3) {$m$};
\node[above] at (0.25,4) {\small $1$};
\node[above] at (0.75,4) {\small $2$};
\node[above] at (1.25,4) {$\ldots$};

\foreach \x in {1,2,...,6} {
	\draw[thick] (\x*0.5-0.5,4)--(1.5+\x*0.5,2);
}

\node[rectangle, below, fill=white] at (3,1.8) {diagonals}; 
\end{tikzpicture}
\end{center}
\caption{The adjacency matrix of the matching graph $M$. We show diagonals (solid lines) and a quadruple of related cells (black). Note that among any quadruple of related cells, only one can belong to a diagonal.}
\label{fig:matrix}
\end{figure}

We first use an adversary-based argument to show an $\Omega(\mathcal{S})$ time lower bound.

\begin{restatable}{lemma-restatable}{deterministiclb}\label{lm:det_LB}
Any deterministic algorithm for \GPM requires $\Omega(\mathcal{S})$ time.
\end{restatable}

\begin{proof}
We will show that any deterministic algorithm checking if there exists at least one occurrence needs to inspect $\Omega(\S)$ entries of $M$ in the worst case by an adversary-based argument. In particular, this implies a lower bound of $\Omega(nm)$ when $\S=\Theta(nm)$. The main difficulty in the argument is to design the input so that the second oracle is essentially useless.

It will be convenient for us to think in terms of the adjacency matrix of the matching graph~$M$ that we denote by $\mathcal{M}$.
Let us assume that $n\geq 2m$ is even, $\Sigma_{T}=[n]$, and $\Sigma_{P}=[2m]$. We split both alphabets into halves. For every $a\in [n/2]$ and $b\in [m]$ we will choose one of the following two possibilities:

\begin{enumerate}
\item $\mathcal{M}[a,b]=\mathcal{M}[n/2+a,m+b]=1$ and $\mathcal{M}[n/2+a,b]=\mathcal{M}[a,m+b]=0$,
\item $\mathcal{M}[a,b]=\mathcal{M}[n/2+a,m+b]=0$ and $\mathcal{M}[n/2+a,b]=\mathcal{M}[a,m+b]=1$.
\end{enumerate}

We call $\mathcal{M}[a,b]$, $\mathcal{M}[n/2+a,b]$, $\mathcal{M}[a,m+b]$ and $\mathcal{M}[n/2+a,m+b]$ \emph{related}. Observe that, irrespectively of all such choices, the second oracle returns the same number for every $b\in\Sigma_{P}$ and every $a\in \Sigma_{T}$,
and so the algorithm only needs to query the first oracle. 

We choose the text $T= 1 \ 2 \ldots n/2 \ 1 \ 2 \ldots n/2$ and the pattern $P = 1 \ 2 \ldots m$. Clearly, $P$ occurs in $T$ when, for some $a \in [n/2]$, we have $M[1+(a+b-2) \bmod n/2, b] = 1$ for every $ b\in [m]$. We call the set of corresponding entries of $M$ a \emph{diagonal} (see Fig.~\ref{fig:matrix}).

Note that among any quadruple of related entries exactly one can belong to the diagonals. Furthermore, suppose that an algorithm retrieves the values in a quadruple of related entries $\mathcal{M}[a,b]$, $\mathcal{M}[n/2+a,m+b]$, $\mathcal{M}[n/2+a,b]$, $\mathcal{M}[a,m+b]$. This can be done by one of the following queries: ask for the value of any of these four entries, or retrieve the particular neighbor of one of the nodes $a$, $n/2+a$, $b$, or $m+b$. In both cases, we retrieve only the related entries and spend $\Omega(1)$ time for any of the retrieved quadruples.

The adversary proceeds as follows. If the algorithm retrieves a quadruple containing $\mathcal{M}[a,b]$, for $a \in [n/2]$ and $b \in [m]$, such that the value of $\mathcal{M}[a,b]$ is not yet determined, the adversary checks if setting $\mathcal{M}[a,b] = 1$ would result in creating a diagonal containing only 1s. If so, the adversary sets $\mathcal{M}[a,b]=0$, and otherwise the adversary sets $\mathcal{M}[a,b] = 1$. In other words, the adversary sets $M[a,b] = 0$ when it is the last undecided entry on its diagonal. 

The algorithm can report an occurrence only after having verified that the corresponding diagonal contains only 1s, and the adversary makes sure that this is never the case. On the other hand, if the algorithm terminates without having reported an occurrence while there exists a diagonal that has not been fully verified then the adversary could set its remaining entries to 1s and obtain an instance that does contain an occurrence. Consequently, the algorithm needs to retrieve all the entries in all the diagonals, and as we showed, it requires $\Omega (mn) = \Omega(\mathcal{S})$ time. 

Note that above $\S = nm/2$. The proof can be extended to $\S < nm/2$ as follows. If $\S \geq m$ we set $n' = \lfloor \S / m \rfloor$ and choose the text to be the prefix of length $n$ of $(1 \ 2 \ldots n')^{\infty}$ (the string $1 \ 2 \ldots n'$ repeated infinitely many times). Then the above argument shows that any algorithm needs to inspect $n' m \geq \S/2$ entries of $\mathcal{M}$. If $\S<m$ we choose the pattern to be the prefix of length $m$ of $(1 \ 2 \ldots \S)^{\infty}$ (the string $1 \ 2 \ldots \S$ repeated infinitely many times), the text to be $1^n$ ($1$ repeated $n$ times) and proceed as above to argue that one must inspect $\Omega(\mathcal{S})$ entries of~$\mathcal{M}$.
\end{proof}

We now move to Monte Carlo algorithms that determine all occurrences of the pattern in the text with small error probability.

\begin{restatable}{lemma-restatable}{randlb}\label{lm:random_LB}
Any Monte Carlo algorithm for \GPM with constant error probability $\eps < 1/2$ requires $\Omega(\mathcal{S})$ time.
\end{restatable}

\begin{proof}
By Yao's minimax principle~\cite{yao}, we only have to exhibit a distribution on the inputs such that, for any deterministic algorithm $A$ that errs with probability at most $\eps$, the expected time is $\Omega(\mathcal{S})$. We will show that this bound holds even for a simpler problem, when the algorithm needs to return a single bit $\texttt{true}$ if the pattern matches all substrings of the text, and $\texttt{false}$ otherwise.

We will only show the details of the argument for even $n\geq 2m$ and $\S=nm/2$, but it can be generalised as above. We also assume for convenience that $\eps<1/4$, but this can be amplified to $\eps<1/2$ by standard argument. 

The distribution on the inputs is defined as follows. We consider $\mathcal{M}$, $T$ and $P$ as in Lemma~\ref{lm:det_LB}, except that now with probability $p$ to be fixed later the entries on every diagonal of $M$ are set to 1, and with probability $1-p$ exactly one of these entries is set to $0$. In other words, for any input the pattern either matches all substrings of the text, or there is exactly one substring that does not match the pattern due to exactly one character. 

Due to the way we choose the values of the related entries in $\mathcal{M}$ (see Lemma~\ref{lm:det_LB}), the oracle that returns the degrees of the characters in the matching graph is useless. The other two oracles spend $\Omega(1)$ time to retrieve the values of any quadruple of related entries. Any deterministic algorithm $A$ inspects a sequence of $x$ quadruples of related entries of~$\mathcal{M}$. If after inspecting the $i$-th quadruple it detects a $0$ in a diagonal, it returns \texttt{false} and stops. In this case, the algorithm is always correct. If it never detects a $0$, after inspecting all $x$ quadruples, it returns $b \in \{\texttt{false},  \texttt{true}\}$. 
We calculate the error probability for every $x$ and $b$.

\begin{enumerate}
\item If $b=\texttt{false}$ then $A$ errs with probability $p$;
\item If $b=\texttt{true}$ then $A$ errs with probability at least $(1-p) \cdot (\S-x)/\S$.
\end{enumerate}

We choose $p\in (2\eps,1-2\eps)$, which exists for $\eps < 1/4$. Since $A$ has error probability at most~$\eps$, we have $(1-p) \cdot (\S-x)/\S \leq \eps$, so $x \geq  \frac{1-p-\eps}{1-p} \cdot \S$. 
The expected number of quadruples retrieved by the algorithm equals
\[p\cdot x+(1-p)\cdot\left(\sum_{i\in [x]} (i/\S) +x\cdot(\S-x)/\S\right) = x-(1-p)\cdot\frac{x(x-1)}{2\S}\geq x-\frac{x^2}{2\S}\]
Plugging in the smallest possible value of $x=\frac{1-p-\eps}{1-p} \cdot \S$ makes the expected number of quadruples retrieved by the algorithm $A$ to be at least $\frac{1}{2}\S - \S \cdot \frac{\eps^2}{2(1-p)^2}$, which is $\Omega(\mathcal{S})$ as $1-p > 2\eps$.
The time bound follows.
\end{proof}

We now show lower bounds for \GPM conditional on hardness of Boolean matrix multiplication. 

\begin{conjecture}[\cite{AbboudV2014}]
For any $\alpha,\beta,\gamma,\eps > 0$, there is no combinatorial\footnote{It is not clear what combinatorial means precisely, but fast matrix multiplication is definitely non-combinatorial. Arguably neither is FFT used in our algorithms, thus making them non-combinatorial.} algorithm for multiplying two Boolean matrices of size $N^{\alpha} \times N^{\beta}$ and $N^{\beta} \times N^{\gamma}$ in time $\Oh(N^{\alpha+\beta+\gamma-\eps})$.
\end{conjecture}

A simple adaptation of the folklore lower bound for computing the Hamming distance (cf.~\cite{GawrychowskiU18}) yields the following lower bounds.

\begin{restatable}{lemma-restatable}{conditionallbs}\label{lm:cond_LB_s}
For any $\alpha \geq 1$, and $1 \ge \beta, \eps > 0$, there is no combinatorial algorithm that solves \GPM in time $\Oh(\S^{0.5-\eps} n)$, for $n=\Theta(m^{(1+\alpha)/2})$ and $\S=\Theta(m^{\beta})$.
\end{restatable}

\begin{proof}
We show a reduction from Boolean matrix multiplication. Consider a matrix $A$ of size $x \times y$ and a matrix $B$ of size $y \times z$, where $x = N^\alpha$, $y = N^{\beta}$, $z = N$. We transform the matrix $A$ by replacing every $1$ by the number of the column it belongs to and every $0$ by the don't care character $?$. Similarly, we replace each $1$ in $B$ by the number of the row it belongs to and every~$0$ by the don't care character $?$. 

\begin{example}
Consider $A = ((0,0,1),(1,0,1),(0,1,0))$ and $B = ((1,0,1),(0,1,0),(1,1,0))$. After the transform, they become $((?,?,3),(1,?,3),(?,2,?))$ and $((1,?,1),(?,2,?),(3,3,?))$, respectively. 
\end{example}

We define the text $T = ?^{z^2} A_1 ?^{z-y+1} A_2 ?^{z-y+1} \ldots \, ?^{z-y+1} A_x ?^{z^2}$,
where $A_i$ is the $i$-th row of~$A$, and the pattern $P = B_1 ?^{z-y} B_2 ?^{z-y} \ldots \, ?^{z-y} B_z$, where $B_j$ is the $j$-th column of the matrix~$B$. The length of $T$ is $n = 2z^2 + (x-1)(z-y+1) + xy = \Oh(N^{1+\alpha})$, and the length of $P$ is $m = yz+(z-y)(z-1) = \Oh(N^2)$. Next, we define the matching relationship as follows. Every character different than the don't care is defined to match all characters of the alphabet but itself, and the don't care character matches all characters of the alphabet. Consequently, the alphabet has size $y+1$ and the matching relationship matrix contains $\S = \Theta(y^2) = \Theta(N^{2\beta})$ set bits. 

Let $C = A \times B$. By definition, $C[i, j] = 1$ iff $\bigvee_{k=1}^{y} (A_i [k] \wedge B_j[k]) = 1$. We claim that this is the case iff, aligning $A_i$ in the text and $B_j$ in the pattern does not yield an occurrence of the pattern. Suppose first that $\bigvee_{k=1}^{y} (A_i[k] \wedge B_j[k]) = 1$. Then there is $k_0$ such that $A_i[k_0] = B_j[k_0] = 1$. In the text and in the pattern they are both encoded by the same $k_0 \neq \, ?$ and aligned, and $k_0$ does not match itself. Therefore, we do not have an occurrence. Assume otherwise. We need to show that for every character $a \neq \, ?$, $a$ is not aligned with itself. For $B_j$ it follows from the fact that $\bigvee_{k=1}^y (A_i[k] \wedge B_j[k]) \neq 1$. For other columns of $B$ it follows from the shift caused by the don't care characters. 

It follows that a combinatorial algorithm that correctly outputs all occurrences of $P$ in~$T$ in $\Oh(\S^{0.5-\eps} n)$ time implies a combinatorial algorithm for Boolean matrix multiplication of matrices of size $N^\alpha \times N^{\beta}$ and $N^{\beta} \times N$ in time $\Oh(\S^{0.5-\eps} n) =\Oh(N^{1+\alpha+2\beta(0.5 - \eps)}) = \Oh(N^{\alpha+1+\beta-2\eps\beta})$, which contradicts the combinatorial matrix multiplication conjecture.
The lower bound follows. 
\end{proof}

\begin{restatable}{corollary-restatable}{conditionallbdi}\label{cor:cond_LB_d_i}
For any $\alpha \geq 1$, and $1 \ge \beta, \eps > 0$, there is no combinatorial algorithm that solves \GPM in time $\Oh(\D^{1-\eps} n)$, for $n=\Theta(m^{(1+\alpha)/2})$ and $\D=\Theta(m^{\beta})$. For any $\alpha \geq 1$, and $1 \ge \eps > 0$, there is no combinatorial algorithm that solves \GPM in time $\Oh(\I^{0.5-\eps} n)$, for $n=\Theta(m^{(1+\alpha)/2})$ and $\I=\Theta(m)$. 
\end{restatable}
\begin{proof}
To show the first part of the claim, note that in the constructed instance of generalized pattern matching $\D = \Theta (m^{\beta/2})$. For the second part, we take $\beta = 1$. Then $\I = \Oh(m)$, and therefore a combinatorial algorithm that correctly outputs all occurrences of $P$ in~$T$ in $\Oh(\I^{0.5-\eps} n)$ time implies a combinatorial algorithm for Boolean matrix multiplication of matrices of size $N^\alpha \times N$ and $N \times N$ in time $\Oh(\I^{0.5-\eps} n) =\Oh(N^{1+\alpha+2(0.5 - \eps)}) = \Oh(N^{\alpha+2-2\eps})$, which contradicts the combinatorial matrix multiplication conjecture.
\end{proof}

\bibliographystyle{plainurl}
\bibliography{main}

\end{document}